\documentclass[12pt,english]{article}
\usepackage[T1]{fontenc}
\usepackage[latin9]{inputenc}
\usepackage{stackrel,amsmath,amsthm,cleveref}
\usepackage[dvipsnames]{xcolor}
\usepackage{amsmath}
\usepackage{esint}
\usepackage{babel}
\usepackage{ifsym}
\usepackage{amstext}
\usepackage{bbm}
\usepackage[cal=boondoxo]{mathalfa}

\usepackage{geometry}
\usepackage{color}
\usepackage{textcomp}
\usepackage{amssymb}
\usepackage{stackrel}
\usepackage{setspace}
\usepackage{graphicx}
\usepackage[round, sort&compress, authoryear]{natbib}
\usepackage{caption}
\usepackage{subcaption}
\usepackage{booktabs}
\usepackage[subtle]{savetrees}
\usepackage{xurl}

\usepackage{pgfplots}
\usepgfplotslibrary{fillbetween}
\usetikzlibrary{patterns}
\usepackage{bm}
\usetikzlibrary{shapes.geometric}
\pgfplotsset{width=13cm,compat=1.10}

\pgfmathdeclarefunction{gauss}{2}{%
  \pgfmathparse{1/(#2*sqrt(2*pi))*exp(-((x-#1)^2)/(2*#2^2))}%
}
\def\i{\ell}

\newtheorem{theorem}{{\bf \sc Theorem}}
\newtheorem{lemma}{{\bf \sc Lemma}}

\newtheorem{proposition}{{\bf \sc Proposition}}
\newtheorem{remark}{{\bf \sc Remark}}

\crefname{claim}{claim}{claims}

\title{Bertrand Menu Competition\thanks{
We are grateful to Felipe Brugu\'es, Erik Madsen, Ellen Muir, Debraj Ray,  Rakesh Vohra, and especially to Koji Yokote for their helpful comments and conversations, as well as to seminar audiences.  Elizabeth Nanami Aoi, Masato Eguchi, and  Aika Okemoto provided excellent research assistance.  Fuhito Kojima is supported by the JSPS KAKENHI Grant-In-Aid 21H04979 and JST ERATO Grant Number JPMJER2301, Japan. Bobak Pakzad-Hurson acknowledges support from the James M. and Cathleen D. Stone  Inequality Initiative.
}}
\author{Fuhito Kojima\thanks{Department of Economics, The University of Tokyo, and the University of Tokyo Market Design Center (UTMD).
Email: fuhitokojima1979@gmail.com}
\hspace{0.5cm} Bobak Pakzad-Hurson\thanks{Department of Economics, Brown University. Email: bph@brown.edu}}

\begin{document}
\maketitle

\begin{abstract}
    We study a variation of the price competition model a la Bertrand, in which firms must offer menus of contracts that obey monotonicity constraints, e.g., wages that rise with worker productivity to comport with equal pay legislation. While such constraints limit firms' ability to undercut their competitors, we show that Bertrand's classic result still holds: competition drives firm profits to zero and leads to efficient allocations without rationing. Our findings suggest that Bertrand's logic extends to a broader variety of markets, including labor and product markets that are subject to real-world constraints on pricing across workers and products.
\end{abstract}

\newpage

\section{Introduction}

Bertrand's model of simultaneous price setting is a workhorse in pinning down equilibrium outcomes in markets with competing firms. When firms sell homogeneous goods and experience constant marginal cost of production, Bertrand competition between as few as two firms drives prices to marginal cost and profit to zero. Papers in the industrial organization literature \citep[e.g.][]{thomadsen05ownership,EINAV2021389} nest Bertrand competition to determine prices (at marginal cost), papers in the labor literature \citep[e.g.][]{Postel2002,Cahuc2006} nest Bertrand competition to determine wages (at marginal productivity), and papers in the education literature \citep[e.g.][]{NPH} nest Bertrand competition to determine financial aid packages (at marginal university utility). The Bertrand model is so relied upon that its conclusions are often assumed as a primitive; \cite{rothschildstiglitz}, \cite{costrelllourey}, and \cite{azevedo} abstract away from the competitive process and begin with the premise that the presence of---or even the potential entry of---a competitor leads to zero profits.

In this paper, we study whether the conclusions of Bertrand's model hold in settings where firms compete over menus of products, wherein firms face a monotonicity constraint in pricing across products. Many competitive situations feature heterogeneous products---producers sell goods of varying quality, employers hire workers of different productivity, and universities provide aid to students of disparate ability. If firms were able to set prices independently across products on the menu, the standard Bertrand argument would imply that competition drives profit to zero for each product. However, under monotonic pricing, in which more productive workers, higher quality and cost goods, and higher ability students are associated with higher prices, a ``product-by-product'' Bertrand argument does not hold; a firm may be unable to change the price of one product in isolation depending on prices of others.

A monotonicity constraint on prices is realistic to capture important considerations across markets. First, monotonicity may be an exogenously imposed constraint due to legal regulations. For example, \cite{cowgill-PH} characterize equal pay laws as requiring wages within a firm to be monotone non-decreasing in productivity; otherwise, a firm would be in violation of the principle of ``equal pay for equal (or better) work,'' a common legal standard.

Second, there may be endogenous reasons that a monotonicity condition obtains. Consider a model where firms compete only along the price dimension, as in Bertrand's original framework, by simultaneously announcing a menu of prices for different products for consumers to then select among. Consumers have different types corresponding to a different ``satisficing'' level of utility. As examples, consider markets where: patients vary in their medical ailments and (the insurance company of) each patient will optimally select the cheapest intervention that corrects her ailment even if that intervention is costlier to provide for a hospital than another intervention that corrects her ailment; heterogeneous buyers require different manufacturing equipment to construct final goods in their respective industries and each will optimally select the cheapest machine capable of completing its task even if it is costlier to produce for a supplier than another machine capable of constructing the final good; heterogeneous workers with different skills select which job to fill within their employing firms and each will optimally select the highest paying job she has the skills to compete even if it is of lower value to the employer than another job she is qualified for. Monotonic prices prevent adverse selection as described above, and no equilibrium can contain non-monotonicities in the prices of transacted goods---if there were, a firm that transacts according to a non-monotonic menu would have a profitable deviation to raising the price for the higher quality/cost product  to induce a monotonic menu which would in turn induce the consumer to select the ``correct'' product without affecting transfers.\footnote{In addition to preventing adverse selection, monotonicity in pricing may be necessary to prevent moral hazard in some markets. For example, again consider the worker-firm example, but suppose firms have the capability of selecting each worker's role after posting a menu of wages but do not observe each workers' skills.
     It is likely impossible for a worker to convince an employer that she has skills she does not have---she cannot pretend to speak a foreign language she does not know---but she could hide that she can speak another foreign language she does know. Monotonicity in pricing removes incentives for workers to ``shirk'' in this way. A similar case for monotonicity is made to prevent shirking in matching markets---a student whose financial aid package is non-monotonic in her grades may intentionally reduce effort to secure more funding---by  \cite{BalinskiSonmez99} and \cite{sonmez}.} Therefore, it is without loss of generality to impose monotonicity as a constraint to guard against manipulation incentives of consumers.

Formally, we study a model with two homogeneous firms that simultaneously announce menus of prices and quantities  over a continuum of products, subject to an exogenously imposed monotonicity constraint that each firm sets a weakly higher price for ``more valuable'' products. To fix ideas, we henceforth describe our model in the language of constant-returns-to-scale firms that seek to hire workers with different productivities. We allow for endogenous contracting (i.e. firms can elect not to hire  workers of particular productivity levels) and rationing (i.e. firms can hire a strict subset of workers available at any productivity level).

Our main result finds that the set of equilibrium allocations  corresponds exactly to the set of Bertrand allocations: the allocation is efficient (almost all workers are hired) and the wage of almost every worker equals her productivity. Efficiency implies both that firms collectively hire workers of all productivity levels and that all workers of each productivity level are hired, that is, firms do not collectively restrict the set of contracts offered, nor do they ration hiring. 
Our result suggests that a Bertrand allocation obtains in a broader variety of markets, and in particular, that multi-product firms price goods at marginal cost/price under competition, even with relevant constraints on pricing across products.

We prove our result via a comparison to a cooperative, matching-with-wages game in which firms can hire, fire, and poach workers, but cannot unilaterally lower  wages of matched workers. We first show that the set of core allocations of the cooperative game corresponds to the set of Bertrand allocations, and second that the set of core allocations corresponds to the set of equilibrium allocations of our original, non-cooperative game. Therefore, our result suggests that a Bertrand allocation obtains ``in the long run,'' even without the structure of and timing enforced by our non-cooperative game.

Our proof also demonstrates that our finding is robust to the (``endogenous'' and ``exogenous'' monotonicity) market settings described above. Our Bertrand result holds if: firms are only able to announce wages and cannot restrict hiring, firms must abide by monotonicity only for accepted wages rather than for wage offers, and if there are more than two firms. Moreover, our result also holds in the case where there are finitely many different products (worker productivity levels) instead of a continuum of different products; the ``finite case'' is significantly easier to analyze because there are fewer restrictions and complexities implied by the monotonicity constraint, and its proof is available upon request.

Our result contrasts with findings on Bertrand competition with endogenous contracting in the literature. 
\cite{rothschildstiglitz} show no contracting occurs for certain parameter values, while  \cite{azevedo} show, in general, not every contract is offered by firms in equilibrium. These distinctions are driven by adverse selection. As described, our monotonicity constraint guards against selection ``downward'' by workers, whereas the analogous conditions in \citet{rothschildstiglitz} and \citet{azevedo} consider both ``downward'' and ``upward'' selection. With the restriction of monotonic wages, we show that it is always possible for firms to deviate from a non-Bertrand allocation such that wages remain monotone and profit increases; 
our proof considers exhaustive cases, and in each, a firm can hire or fire to increase profit, even when it must ``sacrifice'' by increasing the wages of  productive, underpaid workers to abide by monotonicity.
 
Throughout, we describe our model in the context of two firms competing to hire a continuum of workers by simultaneously making weakly positive wage offers. Instead, we could have described our model in the context of two firms competing to sell differentiated products to a continuum of consumers by simultaneously announcing weakly positive prices. These two settings yield analogous results in equilibrium: in the former, almost every worker is hired at a wage equaling her marginal productivity. In the latter, almost every good is sold at a price equaling its marginal cost.

\Cref{model section} presents our model and main result, \Cref{proof section} presents the proof, and \Cref{robust} discusses robustness. Technical details are relegated to the Appendix.

\section{Model and Result}\label{model section}

There are two firms $1,2$ and a continuum of workers. 
A worker type is identified by a pair  $(v,\i) \in [0,1]^2$, where  $v$ is the worker's productivity, and $\i$ is an index.\footnote{Describing a worker's type as both a productivity and an index will be useful to formally account for the ``share'' of workers of each productivity level who receive job offers from both firms, from just firm 1, from just firm 2, and from neither firms. It will similarly be useful to describe the ``share'' of workers of each productivity level who accept offers from each firm.} We assume there is a non-atomic measure $\mu$ that governs the distribution of worker types, where $f(v)$ denotes the density of workers with productivity $v$. Let $F(\cdot)$ denote the associated cumulative distribution function.\footnote{More formally, we define a Borel measure $\tilde \mu^\text{p}$ such that
 $
 \tilde \mu^\text{p}\big([0,x]\big)=F(x)$  for all $x \in [0,1]$,
which exists and is unique \citep[Proposition 25, Section 20.3]{royden}.  Let $\mu^\text{p}$ be the unique measure defined on the Lebesgue measurable sets and coincides with $\tilde \mu^\text{p}$ on Borel measurable sets: such $\mu^\text{p}$ exists and is unique because of the Caratheodory Extension Theorem and the Hahn Extension Theorem \citep[see][Theorems 7.3 and 7.2']{stokeylucas}. $\mu^{\text{p}}$ is the Lebesgue measure on Lebesgue  sigma-algebra $\mathcal{B}^{\text{p}}$ on $[0,1]$, representing the measure of productivity. Similarly, let $\mu^{\text{w}}$ be a measure on a sigma-algebra $\mathcal{B}^{\text{w}}$ on $[0,1]$, representing the measure of indices. We assume that both $\mu^{\text{p}}$ and $\mu^{\text{w}}$ are non-atomic. We assume that the measure $\mu$ over worker types is given by the product measure of  $\mu^{\text{p}}$ and  $\mu^{\text{w}}$, and the  density function associated with $\mu$  is given by $f(v) \times g(\i)$, where  $f(v)$ is associated with measure $\mu^{\text{p}}$ and represents the density of workers with productivity $v$ while $g(\i)$ is associated with measure $\mu^{\text{w}}$ and represents the density of workers whose indices are $\i$.} We additionally assume  $0<\underline f \le \bar f<+\infty$, where $\underline f:=\inf \big\{f(v)| v \in [0,1]\big\}$ and $\bar f:=\sup \big\{f(v)| v \in [0,1]\big\}$.

The game proceeds as follows. First, each firm $i$ simultaneously announces a measurable set of workers  $\mathcal{S}_i$ to which it makes job offers, as well as a measurable function $\mathcal{w}_i$ on $\mathcal{S}_i$ where $\mathcal{w}_i(v,\i)$ is the wage offered to worker $(v,\i)\in \mathcal{S}_i$. We require wage offers to be non-negative, that is, $\mathcal{w}_i(v,\i)\geq 0$ for all $i$ and $(v,\i).$ 
Second, each worker observes the identity of the firm that made an offer to her (if any) and the associated wage offered to her and chooses to accept one of the offers or stay unassigned and receive the wage of zero. Each firm $i$ is matched to the subset of workers $ S_i\subset\mathcal{S}_i$ who accept its offer, and pays each such worker the offered wage.

Each worker's payoff is equal to her wage if she accepts an offer from a firm and zero otherwise. Firms have a constant-returns-to-scale production technology and seek to maximize profit. Formally, if $S_i$ is  measurable, then firm $i$ obtains payoff
$$
\int_{S_i} \big[v-\mathcal {w}_i(v,\i)\big]d\mu.
$$
If $S_i$ is nonmeasurable, then $i$'s payoff is $-1.$ Our solution concept is pure-strategy subgame perfect Nash equilibrium (``equilibrium'').

The main substantive restriction we make is that each firm's wage offers must be monotone non-decreasing in productivity. Formally, we assume that for any two workers $(v,\i)$ and $(v',\i')$ with $v\geq v'$ who receive wage offers from the same firm $i$, it is the case that  $\mathcal w_i(v,\i) \ge \mathcal w_i(v',\i')$. Note that this implies  all workers of the same productivity who receive offers from the same firm must receive the same wage. Therefore, we consolidate notation hereafter and describe wage offers as a function of productivity:  $\mathcal{w}_i(v)$ represents the wage available to each worker $(v,\i)$ who receives a job offer from $i$.

Each strategy profile induces an \emph{allocation}, which specifies the distribution of workers hired by each firm and the wage paid to each productivity level by each firm.\footnote{We note that a strategy profile species firms' offers and workers' acceptance decisions.} Formally, an allocation for firm $i$ is
 $A_i:=\big\{(f_i(v),w_i(v))\big\}_{v \in [0,1]}$, where $f_i(\cdot)$ and $w_i(\cdot)$ are measurable functions such that:
\begin{enumerate}
    \item $f_i(v) \in [0,f(v)]$ is the density of workers of productivity $v$ hired by $i$,
    \item $w_i(v) \in [0,\infty)$ is the wage $i$ pays to each worker of productivity $v$ it hires, and 
    \item If $f_i(v)=0$, then we fix $w_i(v)=0$.
\end{enumerate}

An allocation is a tuple $A:=(A_1,A_2)$ where $A_i$ is an allocation for firm $i$ such that $f_1(v)+f_2(v) \le f(v)$ for each $v$, that is, total employment does not exceed the supply of workers (a feasibility requirement). Note that $w_i(\cdot)$ refers to the wages \emph{paid} in an allocation, while our earlier notation $\mathcal{w}_i(\cdot)$ refers to the wages \emph{offered}.

An allocation $A$ is a \emph{Bertrand allocation} if for all  $i\in \{1,2\}$ and almost all $v\in[0,1]$: $f_1(v)+f_2(v) = f(v)$, and $w_i(v)=v$ if $f_i(v)>0.$\footnote{Throughout, we write ``almost all'' to mean  ``except for a measure-zero set,'' as is commonly used in measure theory. That is, the stated properties hold for all $v\in [0,1]\setminus V$, where $V$ is a zero (Lebesgue) measure subset of $[0,1].$} Clearly, the set of Bertrand allocations is non-empty; consider the allocation in which firm 1 employs all workers and pays each worker a wage equal to her productivity (for all $v$, $f_1(v) = f(v)$, $w_1(v)=v$, $f_2(v)=0$, and $w_2(v)=0$).

The following is our main result:

\begin{theorem}
An allocation can be induced by an equilibrium if and only if it is a Bertrand allocation.
\label{no_EPL_equilibrium_prop}\end{theorem}

\section{Proof}\label{proof section}
We prove \Cref{no_EPL_equilibrium_prop} in two steps. We first demonstrate that the set of Bertrand allocations is equivalent to the set of core allocations of a cooperative version of our game (which will be formally defined below), and then we  show that 
an allocation  
can be induced by an equilibrium of our non-cooperative game if and only if it is a core allocation of the cooperative game. This approach both aids in exposition (because the cooperative game considers only final allocations, and not strategies) and further generalizes our main finding on the focality of Bertrand allocations to a cooperative setting.

\subsection*{Step 1: Core of a Cooperative Game}

Here, we describe the cooperative game, define the core, and characterize the set of core allocations.

 Consider a cooperative game consisting of the same sets of players as the original non-cooperative game, where the distribution of worker types remains the same. The definition of an allocation is also the same as before. Corresponding to the monotonicity requirement of wages in the noncooperative games, we assume that, for any allocation, $w_i(v)\geq w_i(v')$ if $v\geq v'$ and $f_i(v)>0$.

Under an allocation for firm $i$, $A_i:=\big\{(f_i(v),w_i(v))\big\}_{v \in [0,1]}$, $i$ receives profit 
$$
\pi^{A_i}_i:=\int_0^1 \big[v-w_i(v)\big]f_i(v)dv.
$$

An allocation is \emph{blocked} by a set of workers and a firm if there is an alternative wage schedule for a subset of workers such that both the firm and each worker in the subset obtain a higher payoff than in the present allocation. Formally,
we say that 
an allocation $A:=\big\{(f_i(v),w_i(v))\big\}_{v \in [0,1], i =1,2}$ is \emph{blocked} by firm $j$ via an alternative allocation (for $j$)    $\tilde A_j:=\big\{(\tilde f_j(v),\tilde w_j(v))\big\}_{v \in [0,1]}$ if $\pi^{\tilde A_j}_j > \pi_j^{A_j}$ and, for almost all $v \in [0,1]$, one of the following conditions hold (note that, because we define $\tilde A_j$ to be an allocation, it must satisfy all restrictions imposed on an allocation in addition to those listed below):
\begin{enumerate}
\item $\tilde w_j(v) \ge  w_j(v)$ and $\tilde w_j(v) >  w_{-j}(v)$,
\item $\tilde w_j(v) \ge  w_j(v)$ and $\tilde f_j(v)+f_{-j}(v) \le f(v)$, 
\item $\tilde w_j(v) >  w_{-j}(v)$ and $\tilde f_j(v)+f_{j}(v) \le f(v)$, or
\item $\tilde f_j(v)+f_j(v)+f_{-j}(v) \le f(v)$.
\end{enumerate}

Intuitively, Condition 1 states a ``no wage cuts'' requirement; if a firm $j$ weakly raises the wages of all workers involved, and strictly raises wages for workers employed by the other firm, then these workers are all willing to work for $j$. Condition 2 considers the case in which firm $j$ does not need to poach workers from firm $-j$ to construct the blocking allocation, so the only constraint on  wages is that existing workers' wages are not reduced. Condition 3 considers the case in which firm $j$ does not need to keep any existing workers to construct the blocking allocation, so the only restriction  on  wages is that the wage paid to poached workers is higher than those paid by $-j$ to the same workers. Condition 4 considers the case in which firm $j$ can hire from unemployed workers to construct the blocking allocation, so there is no restriction on the wage for these workers.

An allocation $A$ is said to be a \emph{core allocation} if there exists no firm $j$ and   alternative allocation $\tilde A_j$ for  $j$ that blocks $A.$ 

\begin{proposition}\label{core and Bertrand result}
    In the cooperative game, the set of core allocations coincides with the set of Bertrand allocations. 
\end{proposition}

It is straightforward to show that any Bertrand allocation is a core allocation: no firm $j$ can hire unemployed workers (since $f_1(v)+f_{2}(v)=f(v)$ for almost all $v$) nor can it poach workers from the competing firm without incurring a loss (since $w_{-j}(v)=v$ for almost all $v$ such that $f_{-j}(v)>0$, and Conditions 1 and 3 of the definition of block require poached workers to earn strictly more than they were at firm $-j$). 

 It is more complicated to show that any core allocation is a Bertrand allocation. To do so, we consider any non-Bertrand allocation which necessarily must contain a positive measure set of worker productivities $V$ which satisfy at least one of the following six conditions: (i) workers with $v\in V$ are paid more than their productivities, (ii) workers with $v\in V$ are not all hired and those who are receive pay strictly less than their productivities, (iii) workers with $v\in V$ are not all hired and those who are receive pay equal to their productivities,  (iv) workers with $v\in V$ are all hired by the same firm and are paid strictly less than their productivities, (v) some workers of each productivity $v\in V$ are hired by each firm and none receive pay equal to productivity, and (vi) some workers of each productivity $v\in V$ are hired by each firm and only one firm pays these workers wages equal to productivity. In each of these cases, we construct a blocking coalition.

Even though the above six cases may appear out of blue at first glance, this division follows an interpretable economic reasoning as follows. The first case implies an \emph{individual rationality condition} for firms; almost no worker is paid more than her productivity, for otherwise the employing firm could fire such workers and increase its profit. The second and third cases combine to show that there are no demand restrictions; if the firms collectively do not hire all of the workers, one of the firms can increase its profit by hiring unemployed workers at the lowest possible wage that does not violate the monotonicity constraint. The fourth and fifth cases combine to show that  at least one of the firms must pay a wages equal to productivity for almost all $v$, for otherwise one of the firms can ``out compete'' the other for such workers. The final case shows that it cannot be the case that only one firm pays wages equal to productivity but fails to hire all such workers, for otherwise this firm could fire the workers it employs (from whom it earns zero profit) and ``poach'' workers of the same productivity from the other firm at marginally higher wages.

Of course, the main complication throughout our proof is that the blocking coalitions must be constructed to preserve wage monotonicity. To broadly demonstrate our approach, we now present the argument for the third case, which is relatively simple but showcases some of the main techniques used in the more involved cases.

\begin{remark}\label{case 3 remark}  Consider any allocation $A$ in which $w_i(v)\leq v$ for all $v$ and all $i$, and in which there exists a firm $j$ and a subset of productivities $V$ with positive (Lebesgue) measure such that 
 $f_1(v)+f_2(v) < f(v)$ and $w_j(v)= v$ for all $v\in V$. Then $A$ is not a core allocation.
 \end{remark}

 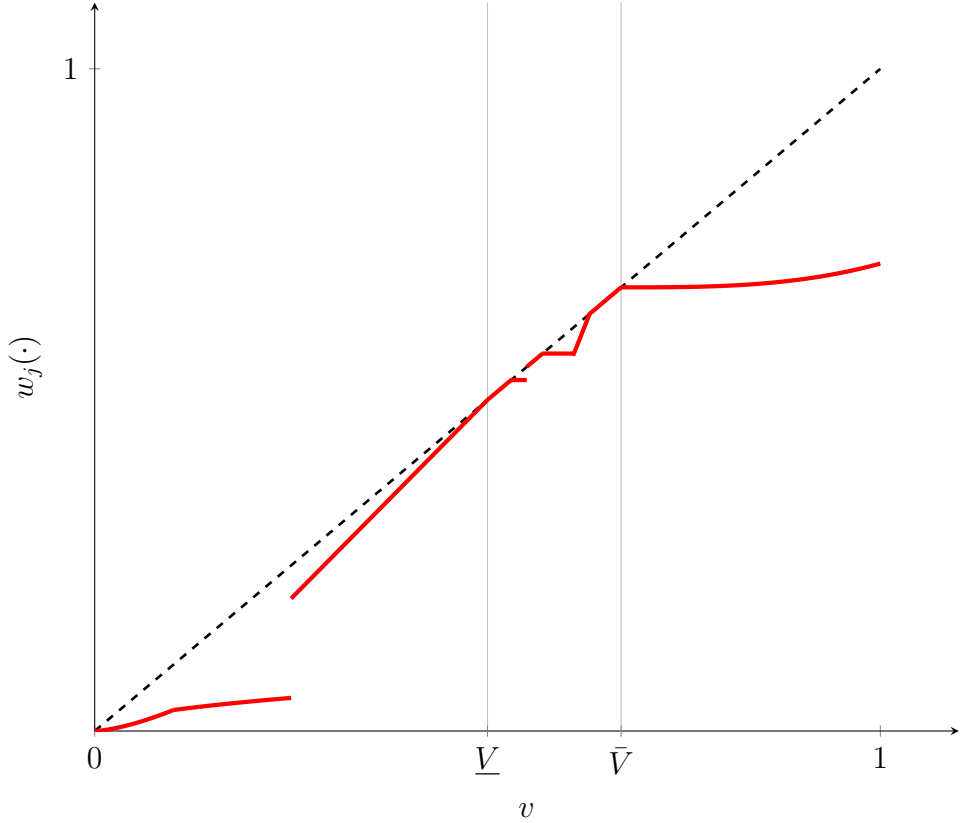
\begin{figure}
     \centering
\begin{tikzpicture}
\begin{axis}[
    xtick={0,1},
    ytick={1},
    ymax=1.1,
    xmax=1.1,
    axis lines=left,
    xlabel={$v$},
    extra x ticks={.5,.67}, extra x tick labels={$\underline V$, $\bar V$},
extra x tick style={grid=major},
    ylabel={$w_j(\cdot)$},
    legend style={at={(1.7,.82)},anchor=west},
]
    \addplot[name path=B, dashed, line width=1pt, domain=0:1] {x};
    
    \addplot[red, ultra thick, domain=0:0.1] {x^(1.5)};          
       \addplot[red, ultra thick, domain=.1:0.25] {.1*x^.5};  
    \addplot[red, ultra thick, domain=0.25:0.5] {.2+1.2*(x-.25)};    
    \addplot[red, ultra thick, domain=0.5:0.53] {x};        
    \addplot[red, ultra thick, domain=0.529:0.55] {.53};  
    \addplot[red, ultra thick, domain=0.55:0.57] {x};   
    \addplot[red, ultra thick, domain=0.569:0.61] {.57};  
    \addplot[red, ultra thick, domain=0.609:0.63] {3*x-1.26};
        \addplot[red, ultra thick, domain=0.629:0.67] {x};
    \addplot[red, ultra thick, domain=0.67:1] {0.67+(x-.67)^3};      
\end{axis}
\end{tikzpicture}
     \caption{Wage function described in \Cref{case 3 remark}}
     \label{case 3 figure}
 \end{figure}

The red curve in \Cref{case 3 figure} depicts a wage function satisfying the conditions in \Cref{case 3 remark}, where $V\subset[\underline V, \bar V]$. Naively, firm $j$ may consider firing all workers in set $V$ who earn their full productivity in wage, and replacing them with workers of the same productivity who were previously unemployed at the minimum wage of 0 (i.e. a blocking outcome $\tilde A_j$ such that $\tilde w_j(v)=0$ and $\tilde f_j(v)=f(v)-f_1(v)-f_2(v)$ for all $v\in V$). But this is not possible, because of the monotonicity constraint, without also firing all workers with productivity strictly less than $\underline V$, which could possibly result in lower overall profits. Instead, could the firm fire all existing employees with productivity  $v\in [\underline V, \bar V]$ and replace them with workers of the same productivity by paying a wage equal to $\underline V$ (i.e. a blocking outcome $\tilde A_j$ such that $\tilde w_j(v)=\underline V$ and $\tilde f_j(v)=f(v)-f_1(v)-f_2(v)$ for all $v\in [\underline V, \bar V]$)? Doing so would satisfy monotonicity, but would also result in the loss of profit from workers of productivity $v\in [\underline V, \bar V]$ who earn strictly less than their productivity. Indeed, such losses in profit may be unavoidable, as there may not exist an interval of productivities such that all workers within the interval earn wages equal to their productivity---recall that $V$ is only required to be measurable, and the measurability of $V$ does not imply the existence of an interval subset of $V$.

 Our proof below constructs a blocking outcome where workers within a certain interval $I:=[\underline v, \bar v]$ of productivities are fired, such that any losses from the fired workers previously earning less than their productivity within this interval are small relative to the gain in profit from ``poaching'' from unemployment all available workers with productivity $v\in I$ and paying them a common wage of no more than $\underline v.$ Our argument, including the proof that an interval $I$ with the desired property exists, is formalized below.

 \begin{proof}
 First, note that there exists $\varepsilon>0$ and $V'\subset [0,1]$ with positive measure such that $f_1(v)+f_2(v) < f(v)-\varepsilon$ and $w_j(v)= v$ for all $v \in V'$.\footnote{\label{proof_footnote}The proof for this claim is as follows: Suppose for contradiction that for each $\varepsilon$, any set of productivities such that $f_1(v)+f_2(v) < f(v)-\varepsilon$ and $w_j  (v)= v$ has zero measure. Then, for each $n =1,2,\dots$ define the set $V_n:=\{v \in V|   f (v)-f_1(v)-f_2(v)>\frac{1}{n} \text{ and } w_j  (v)= v$\}. Then, by assumption,   $V_1,V_2,...$ is an increasing sequence of sets and $\bigcup_{n} V_n=V^*:=\{v \in V|   f (v)-f_1(v)-f_2(v)>0 \text{ and } w_j  (v)= v\}$. Therefore, by countable additivity of the Lebesgue measure, we have $\mu (V^*)=\lim_n \mu (V_n)=0$, which contradicts the assumption that $V$ has positive measure and the fact that $V \subseteq V^*$.} 
 Now, arbitrarily fix $p<1$ such that $(1-p)\bar f \le p \varepsilon$ and $\frac{1}{2} p^2 \varepsilon>   (1-p)\bar f$ (note that those inequalities are satisfied by any sufficiently large  $p<1$). By \citet[Theorem A, Page 68]{Halmos}, there exists an interval $I:=[\underline v,\bar v] \subseteq [0,1]$ such that $\mu(V' \cap I)>p \mu(I)$, where $\mu(\cdot)$ is the Lebesgue measure. Then $\tilde A_j$ where for all $v$:

\begin{align*}
\hspace{-2mm}
\tilde f_j(v) & := \begin{cases} 
 f (v)-f_j(v)-f_{-j}(v) & \text{ if } v \in I,  \\
 f_j(v) & \text{ otherwise.}
\end{cases}
~~~~~~ \tilde w_j  (v) := \begin{cases} 
0 & \text{if } \tilde f_j(v)=0, \\
\underset{v' < \underline v}{\sup~} w_{j}  (v') & \text{if } v \in I \text{ and } \tilde f_j(v)>0,  \\
w_{j}  (v) & \text{otherwise.}
\end{cases}
\end{align*}

\noindent blocks $A$. To see this, note Condition 4 of the definition of block is satisfied for all $v\in I$ (firm $j$ fires all existing workers in set $I$ and hires unemployed workers of the same productivity), and Condition 2 of the definition of block is satisfied for all $v\notin I$ (no worker receives a wage cut and no workers are poached from firm $-j$). By construction, $\tilde w_j  (v)$ satisfies our monotonicity condition. Therefore, it remains only to show that $j$'s profit is higher under $\tilde A_j$.

To see that firm $j$'s profit increases, let $\delta:=\underline v-\sup_{v \le \underline v} w  _j(v)$. $j$ makes an additional profit of at least $$\delta p \mu(I) \varepsilon+ \frac{1}{2} \big(p\mu(I)\big)^2 \varepsilon,$$  from hiring workers from set $V' \cap I$ while firing existing workers from $V' \cap I$ causes no loss (because those workers were hired at wages equal to their productivities), and the loss from losing workers from $I \setminus V'$ is bounded from above by $\big[(1-p) \mu(I) \times (\delta + \mu(I))\big]\bar f=  (1-p)\bar f \delta \mu(I)+(1-p)\bar f \mu(I)^2$.  Because $p$ satisfies $(1-p)\bar f \le p \varepsilon$ and $\frac{1}{2} p^2 \varepsilon>   (1-p)\bar f$ by assumption, the total change in $j$'s payoff is strictly positive, as desired.
\end{proof}

As previously mentioned, the remaining cases are addressed in the appendix.

\subsection*{Step 2: Equivalence between Core and Equilibrium Allocations}

Next, we show that the set of core allocations of the cooperative game is equivalent to the set of equilibrium allocations of the non-cooperative game.

\begin{proposition}\label{equivalence result}
    An allocation $A$ can be induced by an equilibrium of the non-cooperative game if and only if it is a core allocation of the cooperative game.
\end{proposition}
Recall that \Cref{core and Bertrand result} finds that the core of the cooperative game coincides with the set of Bertrand allocations. Therefore, it suffices to show that the set of equilibrium allocations in the non-cooperative game coincides with the set of Bertrand allocations. Without the wage monotonicity constraint, it would be straightforward to establish this result, adopting the original Bertrand argument. 
However, the  monotonicity constraint makes it harder to establish that non-Bertrand allocations cannot be induced by an equilibrium. Specifically, one challenge is that our non-cooperative game requires monotonicity  of wage offers, while the cooperative game imposes monotonicity  on the final wage schedule. The proof below establishes that one can construct a profitable deviation strategy based on a block in the cooperative game which  satisfies the original monotonicity condition.
\begin{proof}
It is easy to see that any Bertrand allocation can be induced by an equilibrium: if both firms offer all workers wages equal to their productivity, then no deviation can increase either firm's profits, given the strategy profile of other agents, as the only way to hire additional workers is to pay them more than their productivity.

To show the complementary direction, consider a strategy profile in the non-cooperative game that induces a non-Bertrand allocation. One can construct a deviation strategy $\sigma_j$ by some firm $j$ in the non-cooperative game that induces (when all others play their original strategies) allocation $A_j^{\sigma_j}:=\big \{(f^{\sigma_j}_j(v),w^{\sigma_j}_j(v))\big\}_{v \in [0,1]}$ for firm $j$, where $A_j^{\sigma_j}$ is the blocking allocation constructed in one of the six cases presented in the proof of \Cref{core and Bertrand result}. Formally, this follows from minor modifications of the proofs of Propositions 7 and 8 in \cite{Gentile2024}. However, recall that the monotonicity constraint in the cooperative game only binds on the final allocation, not on the initial offers; that is, $w^{\sigma_j}_j(\cdot)$---the wages of \emph{employed} workers at firm $j$ according to allocation $A_j^{\sigma_j}$---satisfy monotonicity. Therefore, the remaining concern is that the \emph{offered} wages according to $\sigma_j$ may not be monotonic. If offered wages according to $\sigma_j$ are non-monotonic, consider an alternative deviation strategy  $\sigma'_j$ for firm $j$ that coincides with $\sigma_j$ except that $j$ simply withholds making offers to workers $(v,\i)$ who receive offers from firm $j$ according to $\sigma_j$ but decline it (according to the specified worker strategies). By construction, this again induces the same allocation, i.e.  $A_j^{\sigma'_j}:=\big\{(f^{\sigma'_j}_j(v),w^{\sigma'_j}_j(v))\big\}_{v \in [0,1]}=A_j^{\sigma_j}$. Therefore, for any $v$ such that $f^{\sigma_j}_j(v)=f^{\sigma'_j}(v)>0$, $\mathcal{w}^{\sigma'_j}_j(v)=w^{\sigma'_j}_j(v)=w^{\sigma_j}_j(v)$, where the first equality follows because no offers are rejected, and the second inequality follows because $\sigma'_j$ induces the same allocation as does $\sigma_j$. Because $w^{\sigma_j}_j(\cdot)$ is monotonic, so is $\mathcal{w}^{\sigma'_j}_j(\cdot).$

\end{proof}

For clarity, we construct  a deviation strategy corresponding to the case presented in \Cref{case 3 remark}: Consider a deviation $\sigma'_j$ 
by firm $j$ such that it makes offers to workers in set
$\left(\big[ I^\complement \times [0,1] \big] \cap S_{j}\right) \cup  \left(\big[ I \times [0,1]\big] \cap \big[S_{j} \cup S_{-j}\big]^\complement\right)$, and the wage offer is given by $\tilde w_j(\cdot)$. This deviation clearly results in allocation $\tilde A_j$ for firm $j$, and the set of workers hired is  measurable.\footnote{ To see that the set of workers hired is measurable, note that $I$ is measurable by assumption, and taking a countable number of unions, intersections, and complements of measurable sets results in a measurable set.} Therefore,  $j$'s profit strictly increases, as shown in the proof of \Cref{case 3 remark}, proving the original allocation $A$ cannot be induced by an equilibrium.

\section{Discussion}\label{robust}
We show that multi-product ``menu'' pricing under a monotonicity constraint leads to the familiar Bertrand outcome that competition erodes firm profits. As described in the Introduction, monotone pricing is a relevant constraint in certain markets affected by adverse selection, moral hazard, and legal constraints. Therefore, our paper supports the common, often unmodeled, assumption that prices must equal marginal cost (or marginal productivity in labor-market settings) under competition.

Our finding is robust to a number of model alterations:
\begin{itemize}
    \item Our model allows firms to restrict the set of worker types who receive wage offers, and to ``ration'' wage offers to different types. The proof of \Cref{core and Bertrand result} obtains if we assume firms are required to offer all workers contracts, implying that our result holds if firms can only compete over wages.
    \item \Cref{core and Bertrand result} shows that our main result also holds in a cooperative game,  suggesting that details of offer timing do not drive our finding.
    \item The central contribution of \Cref{equivalence result} is to show that we can obtain any allocation with monotonic \emph{accepted} wages via a strategy that makes monotonic wage \emph{offers}. Because only workers who receive wage offers are eligible to work at that firm, if a firm's wage offers are monotonic, so too are the accepted offers. Therefore, our main result holds if we instead require wage monotonicity among employed workers. 
    \item Our result extends in the usual way when there are more than two firms.    
\end{itemize}
Overall, our main result---together with its robustness to various modeling specifications---suggests that the conclusion of the original Bertrand model holds in a wide variety of environments. We therefore view this paper as providing justification for the pervasive use of the Bertrand prediction as a building block for richer models designed to answer new economic questions.

\newpage
\appendix

\section*{Appendix:
Proof of \Cref{core and Bertrand result}}

\begin{proof}

Consider any Bertrand allocation $A=\big\{(f_i(v),w_i(v))\big\}_{v \in [0,1], i =1,2}$. Then the following hold for almost all $v\in[0,1]$:
\begin{enumerate}
    \item[B1] $f_1(v)+f_2(v) = f(v)$, and 
    \item[B2] for all $i\in \{1,2\}$, $w_i(v)=v$ if $f_i(v)>0$.
\end{enumerate}
We establish the desired result through two lemmas regarding these enumerated conditions. The first lemma shows that any allocation satisfying B1 and B2 is a core allocation by showing that any potential blocking coalition (i.e. satisfying Conditions 1-4 of the definition of block) fails to raise the profit of the indicated firm. The second lemma shows that any allocation failing to satisfy either B1 or B2 admits a blocking coalition. We do so by considering six exhaustive cases, and in each, we construct a block. 

\begin{lemma}\label{noepsw_lemma1}
Any allocation $A$ satisfying B1 and B2 is a core allocation.
\end{lemma}
\begin{proof}[Proof of \Cref{noepsw_lemma1}]
 Suppose not for the sake of contradiction. Then there are a firm $j$ and a distinct allocation (for firm $j$)  $\tilde A_j:=\big\{(\tilde f_j(v),\tilde w_j(v))\big\}_{v \in [0,1]}$ that blocks $A.$ In order for $\tilde A_j$ to block $A$ it must be that $\pi^{\tilde A_j}_j > \pi_j^{A_j}$. However, 

\begin{align*}
     \pi^{A_j}_j= \int_0^1 \big[v-w_j(v)\big]f_j(v)dv= 0\geq  \int_0^1 \big[v-\tilde w_j(v)\big]\tilde f_j(v)dv=  \pi^{\tilde {A}_j}_j.
    \end{align*}

\noindent The second equality follows because, by the construction of $A$, either $f_j(v)=0$ or $w_j(v)=v$ for almost all $v$, therefore, the integrand almost always equals zero. The inequality follows because of the following exhaustive cases for almost all $v$, corresponding, respectively, to Conditions 1-4 of the definition of block:

\begin{itemize}
\item Suppose $\tilde w_j(v) \ge  w_j(v)$ and $\tilde w_j(v) >  w_{-j}(v)$, then it must be that $\tilde w_j(v) \ge v$ since $\max\{w_j(v),w_{-j}(v)\}=v$, which makes the integrand weakly negative,

\item Suppose $\tilde w_j(v) \ge  w_j(v)$ and $\tilde f_j(v)+f_{-j}(v) \le f(v)$. If $\tilde f_j(v)=0$ then the integrand is weakly negative. If $\tilde f_j(v)>0$ then it must be that $f_{-j}(v) < f(v)$, and by the construction of $A$ that  $f_j(v)+f_{-j}(v) = f(v)$, it must be that $f_j(v)>0.$ Therefore, it must be that $ w_j(v)=v$, and the requirement that $\tilde w_j(v) \ge  w_j(v)$ makes the integrand weakly negative.

\item Suppose $\tilde w_j(v) >  w_{-j}(v)$ and $\tilde f_j(v)+f_{j}(v) \le f(v)$. If $\tilde f_j(v)=0$ then the integrand is weakly negative. If $\tilde f_j(v)>0$ then it must be that $f_{j}(v) < f(v)$, and by the construction of $A$ that  $f_j(v)+f_{-j}(v) = f(v)$, it must be that $f_{-j}(v)>0.$ Therefore, it must be that $w_{-j}(v)=v$, and the requirement that $\tilde w_j(v) >  w_{-j}(v)$ makes the integrand strictly negative.

\item Suppose $\tilde f_j(v)+f_j(v)+f_{-j}(v) \le f (v)$ then it must be that $\tilde f_j(v)=0$ since by the construction of $A$ it is the case that $f_j(v)+f_{-j}(v) = f (v)$. Therefore, the integrand is weakly negative.

\end{itemize}
\noindent $\pi_j^{A_j}\geq\pi^{\tilde A_j}_j$ contradicts the premise that $\pi^{\tilde A_j}_j > \pi_j^{A_j}$. Therefore, $A$ is a core allocation. 

\end{proof}

\begin{lemma}\label{noespw_lemma3}
There exist no core allocations which do not satisfy both B1 and B2.
\end{lemma}
\begin{proof}
 Suppose for contradiction that  there is a core allocation $A=\big\{(f_i(v),w_i(v))\big\}_{v \in [0,1], i =1,2}$ such that there exists a set $V$ with positive (Lebesgue) measure where either B1 or B2 fails for all $v\in V$. We proceed by considering six exhaustive cases: By countable additivity of measures, the set of productivities that fails one of B1 or B2 has positive measure if and only if at least one of the sets in the following six cases has a positive measure. 

First, suppose there exist a firm $j$ and a subset of productivities $V \subset [0,1]$ with positive measure such that $w_j(v)>v$ for all $v \in V$. Then $\tilde A_j$ where for all $v$: 

\begin{align*}
\tilde f_j(v) & := \begin{cases} 
f_j(v) & \text{if }  v\notin V,  \\
0 & \text{if }  v\in V.
\end{cases}
~~~~~~ \tilde w_j(v) := \begin{cases} 
w_j(v) & \text{if }  v\notin V,  \\
0 & \text{if }  v\in V.
\end{cases}
\end{align*}
\noindent blocks $A$ as $j$'s profit increases and Condition 4 of the definition of block is satisfied for all $v\in V$ (i.e. the workers in $V$  are fired) and for all $v\in [0,1]\setminus V$ Condition 2 of the definition of block is satisfied (i.e. there is no change in the hiring or wages of workers in $[0,1]\setminus V$). By construction, $\tilde w_j(v)$ satisfies our monotonicity condition.  

Therefore, we proceed with the assumption that for each firm $j$, $w_j(v)\leq v$ for almost all $v$. 
Second, suppose there exist a firm $j$ and a subset of productivities $V$ with positive measure such that  $f_1(v)+f_2(v) < f(v)$ and $w_j(v)< v$ for all $v\in V$. Then $\tilde A_j$ where for all $v$:

 \begin{align*}
\tilde f_j(v) :=  
f(v)-f_{-j}(v).  
~~~~~~~~~~~\tilde w_j(v) :=  \begin{cases} 
0 & \text{if } \tilde f_j(v)=0, \\
\underset{v' \leq  v}{\sup~} w_{j}(v') &  \text{otherwise.}
\end{cases}
\end{align*}

\noindent blocks $A$ as firm $j$'s profit increases  as some previously unemployed workers are hired at a wage strictly less than their productivity while all existing workers at $j$ continue to be employed at the same wage as before, and Condition 2 of the definition of block is satisfied for all $v\in[0,1]$ (i.e. no worker receives a wage cut and no workers are poached from firm $-j$). By construction, $\tilde w_j(v)$ satisfies monotonicity.

Third,  there exists a firm $j$ and a subset of productivities $V$ with positive measure such that 
 $f_1(v)+f_2(v) < f(v)$ and $w_j(v)= v$ for all $v\in V$. This case has been addressed in \Cref{case 3 remark}.

The previous two cases exhaust the possibility of a core allocation in which  $f_1(v)+f_2(v) < f (v)$ for any subset of productivities with positive measure. Therefore, we proceed with the assumption that $f_1(v)+f_2(v) = f (v)$ for almost all $v$.

Fourth, suppose that there exist $j$ and a set $V$ of productivities with positive measure such that $w_j  (v)< v$ and $f_j(v)= f (v)$ for all $v\in V$. Then, there exists $\varepsilon\in(0,1)$ and $V'$ with positive measure such that $w_j  (v)< v-\varepsilon$ and $f_j(v)= f (v)$ for all $v\in V'$.\footnote{The proof is analogous to that in Footnote \ref{proof_footnote}.} For any $p<1$, by \citet[Theorem A, Page 68]{Halmos}, there exists an interval $I^p:=[\underline v^p,\bar v^p] \subseteq [0,1]$ such that $\mu(V' \cap I^p)>p \mu(I^p)$, where $\mu(\cdot)$ is the Lebesgue measure.
Consider the following cases.
\begin{enumerate}
   \item Suppose that there is no $V^p \subseteq [\bar v^p, 1]$ with positive measure such that $w  _{-j}(v) < w  _j(\bar v^p)$ and $f_{-j} (v)>0$ for all $v\in V^p$ and for all $p$ sufficiently close to 1. 
     Let $w^p:=\sup_{v < \underline v^p} w_{-j}  (v).$ 
     Then for a constant $\varepsilon' \in (0,\varepsilon)$, $\tilde A_{-j}$ where for all $v$:
  
\begin{align*}
\hspace{-20mm}\tilde f_{-j}(v) & := \begin{cases} 
 f (v) & \text{ if } v \in I^p,  \\
 f_{-j}(v) & \text{ otherwise.}
\end{cases}
~~~~~ \tilde w_{-j}  (v)  := \begin{cases}
0 & \text{if } \tilde f_{-j}(v)=0,\\
\max\{w^p,w_{j}  (v)+\varepsilon'\} & \text{if } v \in I^p,  \\
\max\{\sup_{v \le \bar v^p} w_{j}  (v)+\varepsilon',w_{-j}  (v)\} & \text{if } v \in [\bar v^p, 1] \text { and } \tilde f_{-j}(v)>0,  \\
w_{-j}  (v) & \text{otherwise.}
\end{cases}
\end{align*}
 \noindent blocks $A$ for sufficiently small $\varepsilon'$ for the following reasons: First, for all $v\in I^p$, Condition 4 of the definition of block is satisfied, and second, for all $v\notin I^p$, Condition 2 of the definition of block is satisfied. Note also that $\tilde w_j  (v)$ satisfies monotonicity by construction. It therefore remains only to show that $\tilde A_{-j}$ increases the profit of firm $-j$.

We proceed by showing that firm $-j$'s ``gain'' from poaching workers in  $I^p \cap V'$ exceeds the ``loss'' of at most $p$ fraction of workers over interval $I^p$ for sufficiently large $p$. After doing so, we show that the loss in profit resulting from the increased wage to workers with productivities $v\geq \bar v^p$ is arbitrarily small, thus completing the argument that $-j$'s profit increases.

The loss from losing existing  workers in $I^p$ is upper bounded by 
\begin{equation}\label{eq:loss}
    \bar f  (1-p)(\bar v^p-\underline v^p)(\bar v^p-w^p).  
\end{equation}

This follows because, in the worst case, there are at most $p$ fraction of workers in $I^p$ who are lost by firm $-j$, with this $p$ fraction  loaded into the rightmost part of $I^p$.

The gain from poaching workers in $I^p \cap V'$ is at least

\[
\underline f \stackrel[\underline{v}^p]{\bar v^p-(1-p)(\bar v^p-\underline v^p)}{\int}\min\{v-w^p,\varepsilon-\varepsilon'\}dv.\]

Let $v^p:=\max\left\{\min\big\{ w^p+\varepsilon-\varepsilon', \bar v^p-(1-p)(\bar v^p-\underline v^p)\big\},\underline v^p\right\}$. Then we can rewrite the lower bound on the gain as 

\begin{equation*}
\underline f \stackrel[\underline{v}^p]{v^p}{\int}(v-w^p)dv+ \underline f \stackrel[v^p]{\bar v^p-(1-p)(\bar v^p-\underline v^p)}{\int}(\varepsilon-\varepsilon')dv. 
\end{equation*}

We can rewrite this as:

 \begin{equation*}
(v^p-\underline v^p)(\underline v^p-w^p)\underline f +\frac{1}{2}(v^p-\underline v^p)^2\underline f  + (\varepsilon-\varepsilon')\big[\bar v^p-(1-p)(\bar v^p-\underline v^p)-v^p\big]\underline f ,
\end{equation*}

which, because $\varepsilon-\varepsilon'\in (0,1)$ and all bracketed terms are non-negative, is no smaller than 

 \begin{align}
& (\varepsilon-\varepsilon')\frac{1}{2}\underline f \left[(v^p-\underline v^p)(\underline v^p-w^p)+(v^p-\underline v^p)^2+\big[\bar v^p-(1-p)(\bar v^p-\underline v^p)-v^p\big]\right] \notag\\
= & 
(\varepsilon-\varepsilon')\frac{1}{2}\underline f \left[(v^p-\underline v^p)(v^p-w^p)+\big[\bar v^p-(1-p)(\bar v^p-\underline v^p)-v^p\big]\right]& \label{eq:gain}
\end{align}

 and therefore, a lower bound on ``the net gain,'' i.e. \eqref{eq:gain}-\eqref{eq:loss}, equals

\begin{align}
\hspace{-8mm}
     (\varepsilon-\varepsilon')\frac{1}{2}\underline f \left[(v^p-\underline v^p)(v^p-w^p)+\big[\bar v^p-(1-p)\big(\bar v^p-\underline v^p)-v^p\big]\right]-\bar f  (1-p)(\bar v^p-\underline v^p)(\bar v^p-w^p).\label{eq:combined_4.1}
 \end{align}

We now consider \eqref{eq:combined_4.1} in light of the three possible values $v^p$ can take:

 First, suppose that $v^p=w^p+\varepsilon-\varepsilon'.$ 
\eqref{eq:combined_4.1} is proportional to
\begin{align*}
     \frac{(v^p-\underline v^p)(v^p-w^p)+\bar v^p-v^p}{(1-p)(\bar v^p-\underline v^p)}-1-\frac{\bar f  (\bar v^p-w^p)}{(\varepsilon-\varepsilon')\frac{1}{2}\underline f },
 \end{align*}
and
\begin{align}
\hspace{-7mm}
    \frac{(v^p-\underline v^p)(v^p-w^p)+\bar v^p-v^p}{(1-p)(\bar v^p-\underline v^p)}-1-\frac{\bar f  (\bar v^p-w^p)}{(\varepsilon-\varepsilon')\frac{1}{2}\underline f } &= \frac{(v^p-\underline v^p)(v^p-w^p)+\underline v ^p -v^p}{(1-p)(\bar v^p-\underline v^p)}+\frac{1}{1-p}-1-\frac{\bar f  (\bar v^p-w^p)}{(\varepsilon-\varepsilon')\frac{1}{2}\underline f } \notag\\
    &= \frac{(v^p-\underline v^p)\left[(v^p-w^p)-1\right]}{(1-p)(\bar v^p-\underline v^p)}+\frac{1}{1-p}-1-\frac{\bar f  (\bar v^p-w^p)}{(\varepsilon-\varepsilon')\frac{1}{2}\underline f } \notag\\
    & \geq \frac{(v^p-w^p)-1}{(1-p)}+\frac{1}{1-p}-1-\frac{\bar f  (\bar v^p-w^p)}{(\varepsilon-\varepsilon')\frac{1}{2}\underline f } \notag \\
    &= \frac{v^p-w^p}{1-p}-1-\frac{\bar f  (\bar v^p-w^p)}{(\varepsilon-\varepsilon')\frac{1}{2}\underline f } \notag\\
    &=\frac{\varepsilon-\varepsilon'}{1-p}-1-\frac{\bar f  (\bar v^p-w^p)}{(\varepsilon-\varepsilon')\frac{1}{2}\underline f }, \notag
\end{align}
\noindent where the first inequality follows because  $1\geq \bar v^p >\underline v^p$  and $\bar v^p\geq v^p\geq \underline v^p \geq w^p\geq 0$ for all $p$ which implies that $\frac{v^p-\underline v^p}{\bar v^p-\underline v^p}\in[0,1]$ and $v^p-w^p\leq 1$, the final equality follows because $v^p=w^p+\varepsilon-\varepsilon'$.   
Clearly this expression is positive for any sufficiently large $p<1$ since $\varepsilon-\varepsilon'>0$ by assumption.

Second, suppose that $v^p=\bar v^p-(1-p)(\bar v^p-\underline v^p)$. Then \eqref{eq:combined_4.1} is proportional to

$$
    \frac{(\varepsilon-\varepsilon')\frac{1}{2}\underline f }{(1-p)}\frac{p(v^p-w^p)}{\bar v^p-w^p}-\bar f. 
$$

\noindent  
We can see that for all $p$

\begin{align}
    \frac{(\varepsilon-\varepsilon')\frac{1}{2}\underline f }{(1-p)}\frac{p(v^p-w^p)}{\bar v^p-w^p}-\bar f &= \frac{(\varepsilon-\varepsilon')\frac{1}{2}\underline f }{(1-p)}\left[p-p(1-p)\frac{\bar v^p-\underline v^p}{\bar v^p-w^p}\right]-\bar f  \notag \\
    &\geq(\varepsilon-\varepsilon')\frac{1}{2}\underline f \left[\frac{p}{1-p}-p\right]-\bar f  \notag \\
    &= (\varepsilon-\varepsilon')\frac{1}{2}\underline f \frac{p^2}{1-p}-\bar f,  \notag
\end{align}

\noindent where the first equality comes from substituting in $v^p=\bar v^p-(1-p)(\bar v^p-\underline v^p),$ and the inequality follows because $\frac{\bar v^p-\underline v^p}{\bar v^p-w^p}\leq 1$ for all $p$ because $\bar v^p >\underline v^p$  and $\bar v^p\geq v^p\geq \underline v^p \geq w^p$.  
Clearly this expression is positive for any sufficiently large $p<1$ since $\varepsilon-\varepsilon'>0$ by assumption.

Third, suppose that $v^p=\underline v^p.$ Then \eqref{eq:combined_4.1} is proportional to
 \begin{align*}
    (\varepsilon-\varepsilon')\frac{1}{2}\underline f \frac{p}{1-p}-\bar f  (\bar v^p-w^p).
 \end{align*}
Noting that each of the terms in parentheses is non-negative by construction and $\bar v^p>\underline v^p\geq w^p$, then 
the above equation is bounded below by 

\begin{align*}
    (\varepsilon-\varepsilon')\frac{1}{2}\underline f \frac{p}{1-p}-\bar f .
 \end{align*} 

This expression is positive for any sufficiently  large $p<1$ since $\varepsilon-\varepsilon'>0$ by assumption.

Therefore, renormalizing the calculated ``net gain'' term from each of the three possible values $v^p$ can take, we have shown that firm $-j$'s change in profit from workers with $v\in I^p\cap V'$ is at least 
\begin{align}
    \hspace{-20mm}  (1-p)(\bar v^p-\underline v^p) \times 
  \min\left \{(\varepsilon-\varepsilon')\frac{1}{2}\underline f  \left [\frac{\varepsilon-\varepsilon'}{1-p}-1-\frac{\bar f  (\bar v^p-w^p)}{(\varepsilon-\varepsilon')\frac{1}{2}\underline f } \right ], (\bar v^p-w^p) \left [ (\varepsilon-\varepsilon')\frac{1}{2}\underline f \frac{p^2}{1-p}-\bar f  \right], (\varepsilon-\varepsilon')\frac{1}{2}\underline f \frac{p}{1-p}-\bar f \right \},\notag
\end{align}
and this expression is positive for every $\varepsilon' \in (0,\varepsilon)$ and sufficiently large $p<1$. Moreover, it can be observed by inspection that this expression is decreasing in $\varepsilon'.$ Furthermore, the wage paid for workers in $[\bar v^p, 1]$ may increase at most by $\varepsilon',$ resulting in a loss of profit from the increased wage being bounded from above by $\varepsilon'.$ From these observations, for any sufficiently large $p <1$ and sufficiently small $\varepsilon'>0,$ firm $-j$ strictly profits with the block, i.e. $\pi_{-j}^{\tilde A_{-j}}>\pi_{-j}^{A_{-j}}$ as desired.

\item Suppose that there exists a subset of $[0,1)$ whose supremum is $1$ such that, for each $p$ in that subset, there is a set $V^p \subseteq [\bar v^p, 1]$ with positive measure such that $w  _{-j}(v) < w  _j(\bar v^p)$ and $f_{-j} (v)>0$ for all $v\in V^p$. Fix any such $p$ and $V^p$.
 Consider $\tilde A_j$ where for all $v$:
\begin{align*}
\tilde f_j(v) & := \begin{cases} 
f (v)  & \text{if }  v \in V^p,\\
f_j(v) & \text{otherwise.}
\end{cases}
~~~~~~ \tilde w_j  (v) := \begin{cases} 
0 & \text{if } \tilde f_j(v)=0,  \\
w_j  (v) & \text{otherwise.}
\end{cases}
\end{align*}

$\tilde A_j$ blocks $A_j$ for the following reasons: Condition 1 of the definition of block is satisfied for all $v\in V^p$ since $w_{-j}  (v)<w_j  (\bar v^p)\leq \tilde w_j  (v)$ for all $v\in V^p$ by construction, and Condition 2 of the definition of block is satisfied for all $v\notin V^p$. Moreover, firm $j$ obtains a strictly higher profit under this allocation.\footnote{In the proposed block $\tilde A_j$, one could alternatively set $\tilde w_j  (v) := w_j  (v) \text{ for every } v \in [0,1],
$ and the proof works without change.}
\end{enumerate}

Fifth, suppose that there exist $j$ and $V \subset [0,1]$ with positive measure such that $0\le w_{-j}  (v)\leq w_j  (v)<v$ and $f_j(v)\in(0, f (v))$ for all $v \in V$. Then, there exists $\varepsilon>0$ and $V' \subset V$ with positive measure such that  $0\le w_{-j}  (v)\leq w_j  (v)<v-\varepsilon$ and $f_j(v)\in(0, f (v)-\varepsilon)$ for all $v \in V'.$  Then for a constant $\varepsilon'>0$, consider $\tilde A_{j}$ where for all $v$:

\begin{align*}
\tilde f_{j}(v) & := \begin{cases} 
f (v) & \text{ if }  v \in V', \\
f_j(v) & \text{otherwise.}
\end{cases}
~~~~ \tilde w_{j}  (v)  := \begin{cases} 
0 & \text{if } \tilde f_j(v)=0,  \\
w_j  (v)+\varepsilon' & \text{otherwise.}
\end{cases}
\end{align*}
$\tilde A_{j}$ blocks $A_{j}$ for any sufficiently small $\varepsilon'>0$ for the following reasons: Condition 1 of the definition of block is satisfied for all $v\in V'$ since $w_{-j}  (v)\leq w_j  (\bar v)< \tilde w_{j}  (v)$ for all $v\in V'$ by construction, and Condition 2 of the definition of block is satisfied for all $v\notin V'$. To see that firm $j$'s profit increases, first note that $j$ benefits from hiring workers from $V'$, which results in an additional profit of at least $(\varepsilon-\varepsilon')\varepsilon \mu(V').$ Meanwhile, $j$ may lose from paying more for existing workers, but the associated loss is bounded from above by $\varepsilon' \beta.$ Therefore, for any sufficiently small $\varepsilon',$ firm $j$'s profit increases, as desired. Note also that  monotonicity is satisfied by $\tilde w  _j(\cdot)$ because $w  _j(\cdot)$ is monotone and $\varepsilon'$ is a constant.

Cases 4 and 5 exhaust the possibility of a core allocation in which there exists a set $V'$ of positive  measure such that $\max\{w  _1(v),w  _2(v)\}<v$ for almost all $v\in V'$. Therefore, we proceed with the assumption that for almost any $v\in[0,1]$ there exists a firm $j$ such that $w_j(v)=v$.

Sixth, suppose there exist a set $V''$ of positive measure and a firm $j$ such that $0\leq w  _{-j}(v)<w  _j(v)=v$ and $f_{-j}(v)\in(0, f (v))$ for all $v\in V''.$ Intuitively, we proceed by showing that firm $j$ can fire some subset of its workers who receive wages equal to productivity, and poach workers of the same productivity from firm $-j$. We proceed by constructing a set of workers with positive measure where such a maneuver is feasible.

Following earlier arguments, there exist $\delta>0$ and a set $V'$ with $\mu(V')>0$ such that $0\leq w  _{-j}(v)+\delta<w  _j(v)=v$ and $f_{-j}(v)\in(0, f (v))$ for all $v\in V'.$

Let $\mathrm{cl}(V')$ be the closure of $V'$. $\mathrm{cl}(V')$ is compact because it is a closed and bounded subset of $[0,1]$. For any $v \in [0,1]$ and $\varepsilon>0$, define $B_\varepsilon(v):=(v-\varepsilon,v+\varepsilon) \cap [0,1]$ to be the $\varepsilon$-ball around $v.$ 
Consider a collection of sets $\{B_{\varepsilon}(v')\}_{v' \in \mathrm{cl}(V')}$ where $\varepsilon < \frac{\delta}{2}$. It is obvious that $\{B_{\varepsilon}(v')\}_{v' \in \mathrm{cl}(V')}$ covers $\mathrm{cl}(V')$ and, because $\mathrm{cl}(V')$ is compact,  there exist $v_1,v_2,\dots,v_n \in \mathrm{cl}(V')$ such that
$\{B_{\varepsilon}(v_i)\}_{i=1}^n$ covers $\mathrm{cl}(V')$, that is, 
$$
\bigcup_{i=1}^n B_{\varepsilon}(v_i) \supseteq \mathrm{cl}(V').
$$
Therefore, it follows that 
$$
\bigcup_{i=1}^n \big [ B_{\varepsilon}(v_i) \cap V' \big ]= V'.
$$
Because $\mu(V')>0$, this implies that 
$$
\mu\left (\bigcup_{i=1}^n \big [ B_{\varepsilon}(v_i) \cap V' \big ]\right ) >0,$$
so there exists $i \in \{1,\dots,n\}$ such that
$
\mu\left (B_{\varepsilon}(v_i) \cap V'\right ) >0.$

Given the conclusion of the preceding paragraph, fix $i\in\{1,\dots n\}$ such that $
\mu\left (B_{\varepsilon}(v_i) \cap V'\right ) >0.$
 We will show that there exists $v'_i\in B_{\varepsilon}(v_i)\cap V'$ such that $\mu\left ([v_i-\varepsilon,v_i'] \cap V'\right ) >0.$ To see this, suppose not for contradiction. Let $\bar v:=\sup B_{\varepsilon}(v_i) \cap V'$. Take a sequence  $(v^k)_{k=1}^\infty$ such that $v^k \in B_{\varepsilon}(v_i) \cap V'$ for each $k$ and $\lim_{k \to \infty} v^k=\bar v$ (such a sequence exists by definition of $\bar v$.) By the assumption made for the purpose of contradiction, we have that $\mu\left ([v_i-\varepsilon,v^k] \cap V'\right )=0$ for each $k=1,2,\dots$. Since the sets $\left ([v_i-\varepsilon,v^k] \cap V'\right )_{k=1}^\infty$ form an increasing sequence of measurable sets, 
we have $0=\mu\big ([v_i-\varepsilon,\bar v] \cap V'\big)=\mu\big ([v_i-\varepsilon,v_i+\varepsilon] \cap V'\big )=\mu\big (B_{\varepsilon}(v_i) \cap V'\big )>0$, where the inequality is assumed at the beginning of the current paragraph. This is a contradiction.

Therefore, following the preceding paragraph,  fix $v'_i\in B_{\varepsilon}(v_i)\cap V'$ with the property that $\mu\big ([v_i-\varepsilon,v_i'] \cap V'\big) >0.$
Because $v_i' < v_i+\varepsilon$  and $\varepsilon<\frac{\delta}{2}$, we have $[v_i'-\delta, v_i'] \supseteq [v_i-\varepsilon, v_i'].$
Hence, noting that $[v_i'-\delta, v_i'] \cap V'$ and $[v_i-\varepsilon, v_i'] \cap V'$ are measurable, 
$\mu\big([v_i'-\delta, v_i'] \cap V'\big) \geq \mu\big([v_i-\varepsilon, v_i'] \cap V'\big)>0.$

We now show firm $j$ can block allocation $A$ via workers whose productivities fall in $[v_i'-\delta, v_i']$. To do so, we observe that $w  _{-j}(v_i')<v_i'-\delta$ because $v_i' \in V'$. Thus, by the monotonicity of $w  _{-j}$, $w  _{-j}(v)<v'_i-\delta$ for all $v \in [v_i'-\delta,v_i'].$ This implies  $w  _{-j}(v)<v$ for all $v \in [v'_i-\delta, v'_i]$. Therefore, by the ongoing assumption (following the conclusions of Cases 4 and 5) that  $\max\{w  _1(v),w  _2(v)\}=v$ for almost every $v \in [0,1]$, it follows that $w  _j(v)=v$ for almost all $v \in [v_i'-\delta,v_i'].$

Consider $\tilde A_j$ where
 \begin{align*}
 \hspace{-1mm}
\tilde f_{j}(v) & := \begin{cases} 
f_{j}(v) & \text{ if }  v \notin [v_i'-\delta,v_i'], \\
f_{-j}(v) & \text{ if }  v \in [v_i'-\delta,v_i'].
\end{cases} 
~~~~~~~~~~~\tilde w_j  (v) :=  \begin{cases} 
 w_{j}  (v) &  \text{ if } v \notin [v_i'-\delta,v_i'],\\
v_i'-\delta & \text{ if }  v \in [v_i'-\delta,v_i'] \text{ and } \tilde f_{j}(v)>0,\\ 
0 & \text{ otherwise.}
\end{cases}
\end{align*}
 $\tilde A_j$ blocks $A_j$ for the following reasons: First, it is obvious from construction that $\tilde w_j  $ satisfies monotonicity.  Condition 3 of the definition of block is satisfied for all $v\in [v_i'-\delta,v_i']$ (i.e. the workers previously employed by firm $-j$ are successfully poached and some workers are fired), 
and Condition 2 of the definition of block is satisfied for all $v\notin [v_i'-\delta,v_i']$ (i.e. workers in this set do not experience changes to hiring or wages). It is also the case that $\tilde A_j$ provides firm $j$ with higher profit than $A_j$: newly poached workers from $[v_i'-\delta,v_i']\cap V'$
(of whom there are a positive measure) are paid lower wages than their productivity in allocation $\tilde A_j$ while all newly-fired workers are from $[v_i'-\delta,v_i']$ and received wages equal to productivity from $j$ in allocation $A_j$. This shows that $A$ is not a core allocation.

As these six cases are exhaustive and none of them admits a core allocation, we have completed the argument that any core allocation must be a Bertrand allocation. 
\end{proof}
\end{proof}

\newpage

\singlespacing
\bibliographystyle{plainnat}
\bibliography{eplbib}

\begin{thebibliography}{15}
\providecommand{\natexlab}[1]{#1}
\providecommand{\url}[1]{\texttt{#1}}
\expandafter\ifx\csname urlstyle\endcsname\relax
  \providecommand{\doi}[1]{doi: #1}\else
  \providecommand{\doi}{doi: \begingroup \urlstyle{rm}\Url}\fi

\bibitem[Azevedo and Gottlieb(2017)]{azevedo}
Eduardo~M. Azevedo and Daniel Gottlieb.
\newblock {Perfect Competition in Markets With Adverse Selection}.
\newblock \emph{Econometrica}, 85\penalty0 (1):\penalty0 67--105, 2017.

\bibitem[Balinski and S\"{o}nmez(1999)]{BalinskiSonmez99}
Michel Balinski and Tayfun S\"{o}nmez.
\newblock {A Tale of Two Mechanisms: Student Placement}.
\newblock \emph{Journal of Economic Theory}, 84\penalty0 (1):\penalty0 73--94, 1999.

\bibitem[Cahuc et~al.(2006)Cahuc, Postel-Vinay, and Robin]{Cahuc2006}
Pierre Cahuc, Fabien Postel-Vinay, and Jean-Marc Robin.
\newblock {Wage Bargaining with On-the-Job Search: Theory and Evidence}.
\newblock \emph{Econometrica}, 74\penalty0 (2):\penalty0 323--364, 2006.

\bibitem[Costrell and Loury(2004)]{costrelllourey}
Robert~M. Costrell and Glenn~C. Loury.
\newblock {Distribution of Ability and Earnings in a Hierarchical Job Assignment Model}.
\newblock \emph{Journal of Political Economy}, 112\penalty0 (6):\penalty0 1322--1363, 2004.

\bibitem[Cowgill and Pakzad-Hurson(2025)]{cowgill-PH}
Bo~Cowgill and Bobak Pakzad-Hurson.
\newblock {Equal Pay for Equal Work (for Equal Pay)}.
\newblock \emph{mimeo}, 2025.

\bibitem[Einav et~al.(2021)Einav, Finkelstein, and Mahoney]{EINAV2021389}
Liran Einav, Amy Finkelstein, and Neale Mahoney.
\newblock {The IO of selection markets}.
\newblock In Kate Ho, Ali Horta\c{c}su, and Alessandro Lizzeri, editors, \emph{Handbook of Industrial Organization}, volume~5 of \emph{Handbook of Industrial Organization}, pages 389--426. Elsevier, 2021.

\bibitem[Gentile~Passaro et~al.(2024)Gentile~Passaro, Kojima, and Pakzad-Hurson]{Gentile2024}
Diego Gentile~Passaro, Fuhito Kojima, and Bobak Pakzad-Hurson.
\newblock {Equal Pay for \emph{Similar} Work}.
\newblock \emph{Working Paper}, 2024.

\bibitem[Halmos(1974)]{Halmos}
Paul~R. Halmos.
\newblock \emph{{Measure Theory}}.
\newblock Springer-Verlag, 1974.

\bibitem[Nei and Pakzad-Hurson(2021)]{NPH}
Stephen Nei and Bobak Pakzad-Hurson.
\newblock {Strategic Disaggregation in Matching Markets}.
\newblock \emph{Journal of Economic Theory}, 197, 2021.

\bibitem[Postel-Vinay and Robin(2002)]{Postel2002}
Fabien Postel-Vinay and Jean-Marc Robin.
\newblock {Equilibrium Wage Dispersion with Worker and Employer Heterogeneity}.
\newblock \emph{Econometrica}, 70\penalty0 (6):\penalty0 2295--2350, 2002.

\bibitem[Rothschild and Stiglitz(1976)]{rothschildstiglitz}
Michael Rothschild and Joseph Stiglitz.
\newblock {Equilibrium in Competitive Insurance Markets: An Essay on the Economics of Imperfect Information}.
\newblock \emph{The Quarterly Journal of Economics}, 90\penalty0 (4):\penalty0 629--649, 1976.

\bibitem[Royden and Fitzpatrick(2010)]{royden}
H.L. Royden and P.M. Fitzpatrick.
\newblock \emph{Real Analysis (Fourth Edition)}.
\newblock Prentice Hall, 2010.

\bibitem[S\"{o}nmez(2013)]{sonmez}
Tayfun S\"{o}nmez.
\newblock {Bidding for Army Career Specialties: Improving the ROTC Branching Mechanism}.
\newblock \emph{Journal of Political Economy}, 121\penalty0 (1):\penalty0 186--219, 2013.

\bibitem[Stokey and Lucas(1989)]{stokeylucas}
Nancy~L. Stokey and Robert~E. Lucas, Jr.
\newblock \emph{{Recursive methods in economic dynamics}}.
\newblock Harvard University Press, 1989.

\bibitem[Thomadsen(2005)]{thomadsen05ownership}
Raphael Thomadsen.
\newblock {The Effect of Ownership Structure on Prices in Geographically Differentiated Industries}.
\newblock \emph{The RAND Journal of Economics}, 36\penalty0 (4):\penalty0 908--929, 2005.

\end{thebibliography}
\end{document}